\documentclass[conference,a4paper]{IEEEtran}
\usepackage[utf8]{inputenc}
\usepackage[T1]{fontenc}
\usepackage{url}              
\usepackage{cite}             

\usepackage[cmex10]{amsmath}  
\usepackage{mleftright}       
\mleftright                   

\usepackage{graphicx}         
\usepackage{booktabs}         

\usepackage{amssymb}

\usepackage{tikz}

\usepackage{pgfplots}
\usepgfplotslibrary{groupplots,dateplot}

\usepackage{flushend}

\usepackage{bm}
\usepackage{graphicx}         
\usepackage{booktabs}         

\usetikzlibrary{shapes,arrows}
\usetikzlibrary{positioning,fit,calc}
\tikzstyle{block} = [draw, fill=white, rectangle,
minimum height=3em, minimum width=6em]
\tikzstyle{input} = [coordinate]
\tikzstyle{output} = [coordinate]

\newcommand\E{\mathbb{E}}
\renewcommand\P{\mathbb{P}}

\newtheorem{theorem}{Theorem}
\newtheorem{lemma}{Lemma}
\newtheorem{proposition}{Proposition}
\newtheorem{corollary}{Corollary}

\newtheorem{definition}{Definition}
\newtheorem{remark}{Remark}

\definecolor{blueR}{RGB}{0, 0, 0}

\begin{document}

\title{Maximal Leakage of Masked Implementations\\Using Mrs.\ Gerber's Lemma for Min-Entropy}

 \author{
	   \IEEEauthorblockN{Julien Béguinot$^{1}$, Yi Liu$^{1}$, Olivier Rioul$^{1}$, Wei Cheng$^{1,2}$, and Sylvain Guilley$^{1,2}$}
	   \IEEEauthorblockA{$^1$LTCI, Télécom Paris, Institut Polytechnique de Paris, France \\ 
	   					 $^2$Secure-IC S.A.S., France\\
	   					\textit{firstname.lastname}@telecom-paris.fr}
	 }

\maketitle

\begin{abstract}
	A common countermeasure against side-channel attacks on secret key cryptographic implementations is $d$th-order masking, which splits each sensitive variable into $d+1$ random shares.
	In this paper, maximal leakage bounds on the probability of success of any side-channel attack are derived for any masking order. 
	\textcolor{blueR}{Maximal leakage (Sibson's information of order infinity) is evaluated between the sensitive variable and the noisy leakage}, and is related to the conditional ``min-entropy'' (Arimoto's entropy of order infinity) of the sensitive variable given the leakage. The latter conditional entropy is then lower-bounded in terms of the conditional entropies for each share using majorization inequalities.
	This yields a generalization of Mrs. Gerber's lemma for min-entropy in finite Abelian groups.
\end{abstract}

\section{Introduction}
\label{sec:intro}

When a cryptographic device is operating, any kind of physical leakage (time, power, electromagnetic emanations, etc.) can be exploited by an attacker.
The attacker queries the device multiple times, and measures the corresponding leakages to infer the secret key.
The security of devices against side-channel attacks has become a major concern.

To evaluate the probability of success for any side-channel attack, information-theoretic metrics turn out to be effective and have been used in many studies. Using conditional mutual information and Fano's inequality, de Chérisey et al.~\cite{CheriseyGuilleyRioulPiantanida19} established several universal bounds on the probability of success for a given number of queries, or equivalently, the minimum number of queries required to achieve a given level of success.
This approach has been extended to conditional Sibson's $\alpha$-information by Liu et al.~\cite{LCGR21}. However,
both~\cite{CheriseyGuilleyRioulPiantanida19} and~\cite{LCGR21} were restricted to unprotected cryptographic devices.

\emph{Masking} is one of the most well-established countermeasures.
The main issue in this context is the fact that a direct evaluation of the information leakage requires data and computational complexities that increase rapidly with the masking order~\cite{cheng2022attacking}.
Therefore, it is important to derive bounds in terms of the individual information leakages for each share.

Duc et al.~\cite{DucFS15} conjectured a general form of such bounds. 
Rigorous bounds were obtained in two independent recent works by Ito et al.~\cite{DBLP:conf/ccs/ItoUH22} and Masure et al.~\cite{DBLP:journals/iacr/MasureRS22}.
Even more recently, Béguinot et al.~\cite{DBLP:journals/iacr/BeguinotCGLMRS22} improved these results using Mrs. Gerber's lemma~\cite{Wyner1973ATO, Jog2013TheEPI}
 to derive sharp bounds in terms of mutual information for masking in additive groups of order $2^n$.

In the case of unprotected implementations (without masking), it is shown by simulation in~\cite{LCGR21} that the probability of success of a side-channel attack is evaluated using Sibson's $\alpha$-information all the more accurately as $\alpha$ increases. Therefore, the case of mutual information, which corresponds to $\alpha=1$ is not optimal.  This motivates the derivation of new bounds in the limiting case $\alpha=+\infty$.

The usual setup of masking countermeasures involves bitwise XOR (exclusive or) operations, which are particularly well suited to symmetric cryptographic algorithms like AES. However, modern cryptography also relies on operations performed in groups of prime order, and masking can also be multiplicative~\cite{DBLP:conf/ches/AkkarG01} and not only additive~\cite{DBLP:conf/ches/GoubinP99}. For all these reasons, there is a strong incentive to extend the previous bounds to arbitrary finite Abelian groups.
This motivates the generalization of Mrs. Gerber's lemma to any such Abelian group.

Mrs. Gerber's lemma was initially derived by Wyner and Ziv~\cite{Wyner1973ATO} to lower bound the entropy of a modulo 2 addition of binary random variables in terms of the entropies of each summand.  It was extended by Jog and Anatharam~\cite{Jog2013TheEPI} to the case of additive groups of order $2^n$, and  by Hirche~\cite{Hirche2020} to the case of Rényi entropy of binary variables.
The general case of additive groups was only considered by Tao~\cite{Tao2009SumsetAI} for Shannon entropy and independent copies of two shares, in relation to sumset theory.
While the original binary Mrs. Gerber's lemma was used to derive a binary version of the entropy power inequality~\cite{ShamaiWyner90}, a generalization of the entropy power inequality to any prime cyclic additive group and Rényi entropy was investigated by Madiman et al.~\cite{madiman2021entropy}, but does not reduce to an explicit ``Mrs. Gerber's lemma''-type inequality.
Therefore, it appears that the case of min-entropy (Rényi entropy of order $\infty$) and additive groups of any order has not been investigated yet in our context.

\subsection*{Contributions}

In this paper, we show that when evaluating the performance of side-channel attacks of masked implementations using conditional Sibson's $\alpha$-information, the exact performance of optimal maximum likelihood attacks is attained in the limiting case \hbox{$\alpha=+\infty$}.
This motivates the investigation of Mrs. Gerber's lemma for conditional min-entropy (Arimoto's conditional entropy of order $\infty$).
We derive a variation of such Mrs. Gerber's lemma for any finite Abelian group and for any masking order.

The remainder of this paper is organized as follows. Section~\ref{sec-defs} gives some notations and preliminaries on $\alpha$-informational quantities. Section~\ref{sec-opt} shows that the optimal evaluation of side-channel attack success by Fano's inequality is achieved in the limiting case \hbox{$\alpha=+\infty$} and derives the corresponding bound in terms of the information between the sensitive variable and the leakage, which is linear in the number of queries.
Section~\ref{sec-MGL} derives Mrs. Gerber's lemma for min-entropy, first for two summands in any finite Abelian group, then extends it to the general case of $d+1$ summands.
Section~\ref{sec-concl} concludes and gives some perspectives.

\section{Preliminaries and Notations}\label{sec-defs}

\subsection{Framework and Notations}\label{sca-notion}

Let $K$ be the secret key and $T$ be a public variable (usually plaintext or ciphertext) known to the attacker. It is assumed that $T$ is independent of $K$, and $K$ is uniformly distributed over an Abelian group $\mathcal{G}$ of order $M$. The cryptographic algorithm operates on
$K$ and $T$ to compute
a sensitive variable~$X$, which takes values in the same group~$\mathcal{G}$ and is determined by $K$ and~$T$, in such a way that $X$ is also uniformly distributed over~$\mathcal{G}$.

In a masking scheme of order $d$, the sensitive variable~$X$ is randomly split into $d+1$ \emph{shares} $X_0$, $X_1$, \ldots, $X_d$ and cryptographic operations are performed on each share separately. Thus,
$X=X_0 \oplus X_1 \oplus \cdots \oplus X_d$, where each share $X_i$ is a uniformly distributed random variable over $\mathcal{G}$ and $\oplus$ is the group operation in~$\mathcal{G}$. For this group operation, we let $\ominus g$ denote the opposite of $g \in \mathcal{G}$.
A typical example is ``Boolean masking'', for which $\oplus\equiv\ominus$ is the bitwise XOR operation.

During computation, shares $\bm{X}=(X_0, X_1, \ldots, X_d)$ are leaking through some side channel. Noisy ``traces,'' denoted by $\bm{Y}=({Y}_0, {Y}_1, \ldots, {Y}_d)$, are measured by the attacker,
where~$\bm{Y}$ is the output of a memoryless side channel with input $\bm{X}$.
Since masking shares are drawn uniformly and independently, both $\bm{X}$ and $\bm{Y}$ are i.i.d. sequences.
The attacker measures $m$ traces \hbox{$\bm{Y}^m=(\bm{Y}_1,\bm{Y}_2,\ldots,\bm{Y}_m)$}
corresponding to the i.i.d. text sequence \hbox{$T^m=(T_1,T_2,\ldots,T_m)$}, then exploits her knowledge of $\bm{Y}^m$ and $T^m$ to guess the secret key $\hat{K}$.
Again, since the side-channel is memoryless, both $\bm{X}^m$ and $\bm{Y}^m$ are i.i.d. sequences.

Let \hbox{$\P_s=\P(K=\hat{K})$} be the probability of success of the attack upon observing $T^m$ and $\bm{Y}^m$.
In theory, maximum success is obtained by the MAP (maximum a posteriori probability) rule with success probability denoted by $\P_s=\P_s(K|\bm{Y}^m,T^m)$.
The whole process is illustrated
in Fig.~\ref{model1}.

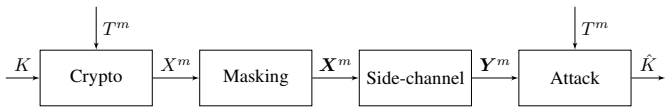
\begin{figure}[h!]
	\centering\vspace*{-1ex}
	\resizebox{\linewidth}{!}{
		\begin{tikzpicture}[auto, node distance=1.7cm,>=latex']
			\node [input, name=input] {};
			\node [block, right of=input] (crypto) {Crypto};
			\node [block, right of=crypto,  node distance=3cm] (masking) {Masking};
			\node [block, right of=masking,  node distance=3cm] (channel) {Side-channel};
			\node [block, right of=channel,  node distance=3cm] (attack) {Attack};
			\node [output, right of=attack] (output) {};
			\node[input] (text) at ([yshift=1.35cm]$(crypto)!0.0!(masking)$) {};
			\node[input] (text2) at ([yshift=1.35cm]$(channel)!1.0!(attack)$) {};
			\draw [->] (crypto) -- node[name=X] {$X^m$} (masking);
			\draw [->] (masking) -- node[name=X] {$\bm{X}^m$} (channel);
			\draw [->] (input) -- node {$K$} (crypto);
			\draw [->] (channel)-- node [name=Y] {${\bm{Y}}^m$}(attack);
			\draw [->] (attack) -- node [name=Khat] {$\hat{K}$}(output);
			\draw[->] (text) -- node{$T^m$} (crypto);
			\draw[->] (text2) -- node{$T^m$} (attack);
	\end{tikzpicture} }
	\caption{Side-channel analysis as a (unintended) ``communication'' channel. {\color{blueR} ``Crypto'' can be any sensitive computation (encryption or decryption). $T$ is a public random variable (e.g., a plain or cipher text byte).}}
	\label{model1}
\end{figure}

\vspace{-1em}
\subsection{R\'enyi's $\alpha$-Entropy and Arimoto's Conditional $\alpha$-Entropy}

Assume that either \hbox{$0<\alpha<1$} or \hbox{$1<\alpha<+\infty$} (the limiting values $0,1,+\infty$ can be obtained by taking limits).  We consider probability distributions $P,Q$ with a dominating measure $\mu$,
with respect to which they follow densities denoted by the corresponding lower-case letters $p,q$.
We follow the notations of~\cite{LCGR21} in the following
\begin{definition}[Rényi $\alpha$-Entropy and $\alpha$-Divergence]
	\begin{align}
		H_\alpha(P)&= \tfrac{\alpha}{1-\alpha} \log\|p\|_\alpha\\
		D_\alpha (P\|Q)&=  \tfrac{1}{\alpha-1} \log\langle p\|q\rangle^\alpha_\alpha \label{divergence}
	\end{align}
	with the special notation:
	\begin{align}
		\|p\|_\alpha &= \bigl(\int |p|^\alpha d\mu \bigr)^{1/\alpha}  \label{alpha-norm}\\
		\langle p\|q\rangle_\alpha &= \bigl(\smash{\int} p^\alpha q^{1-\alpha} d\mu \bigr)^{1/\alpha}.
	\end{align}
\end{definition}
The usual Shannon entropy and Kullback-Leibler divergence are recovered by letting $\alpha\to 1$. The $\alpha$-entropy is nonincreasing in $\alpha$ and achieves its \emph{min-entropy} $H_\infty$ at the limit $\alpha =\infty$:

\begin{definition}[Min-Entropy]
	For a probability distribution $P$ over a finite alphabet, the min-entropy is
	\begin{equation}
		H_\infty(P)=- \log ( \max~ p).
	\end{equation}
\end{definition}

Many different definitions of conditional $\alpha$-entropy $H_\alpha(X|Y)$ were proposed in the literature. We use Arimoto's definition, which is argued to be the most promising one~\cite{FehrBerens14}:
\begin{definition}[Arimoto's Conditional $\alpha$-Entropy~\cite{Arimoto75}]
	The conditional $\alpha$-entropy of $X$ given $Y$ is defined as
	\begin{equation}
		H_\alpha(X|Y) = \frac{\alpha}{1-\alpha}\log \E_Y \|p_{X|Y}\|_\alpha.
	\end{equation}
\end{definition}
Assuming $X$ takes values in a finite alphabet, the conditional min-entropy can be obtained by letting $\alpha \to \infty$ in $H_\alpha(X|Y)$:
\begin{definition}[Conditional Min-Entropy~\cite{teixeira2012}]
	\begin{equation}
		H_\infty(X|Y)= - \log ( \E_Y \max\limits_{x} p_{X|Y} ) = - \log \P_s(X|Y)
	\end{equation}
where $\P_s(X|Y)$ is the maximum average probability of success in estimating $X$ having observed $Y$, by the MAP rule.
\end{definition}

\subsection{Sibson's $\alpha$-Information and Liu et al.'s Conditional Version}

Again, several different definitions of $\alpha$-information $I_\alpha(X;Y)$ have been proposed, and Sibson’s $\alpha$-information is perhaps the most appropriate one because it satisfies several useful properties that other definitions do not~\cite{Verdu15}.
\begin{definition}[Sibson's $\alpha$-Information~\cite{Sibson69,Verdu15}]
	\begin{align}
		I_\alpha(X;Y) &= \min_{Q_Y} D_\alpha(P_{XY}\|P_X \times Q_Y) \\
		&= \tfrac{\alpha}{\alpha-1} \log \E_Y \langle p_{X|Y}\|p_X\rangle_\alpha.
	\end{align}
\end{definition}

\begin{definition}[Max-Information~{\cite[Thm. 4]{HoVerdu15}}]\label{maximal}
	Assuming $X,Y$ are discrete random variables, one has
	\begin{equation}\label{max-i}
		I_\infty (X;Y) = \log \sum_y \sup\limits_{x:p_X(x)>0} p_{Y|X}(y|x).
	\end{equation}
\end{definition}
Max-information is also studied in~\cite{Issa20} as \emph{maximal leakage}.

Again, there are many different proposals for \emph{conditional} $\alpha$-information. We use the following definition which seems most appropriate in the context of side-channel analysis~\cite{LCGR21}:
\begin{definition}[Conditional $\alpha$-Information~\cite{LCGR21}]
	\begin{align}
		I_\alpha(X;Y|Z) &= \min\limits_{Q_{YZ}}D_\alpha(P_{XYZ}\|P_{X|Z}Q_{YZ}) \\
		&= \tfrac{\alpha}{\alpha-1} \log \E_{YZ} \langle p_{X|YZ}\|p_{X|Z}\rangle_\alpha.
	\end{align}
\end{definition}

\bigskip

\section{Fano's Equality for Order $\infty$: Linear Bound}\label{sec-opt}

\subsection{Fano Inequality for Conditional $\alpha$-Information as \hbox{$\alpha\to\infty$}}

Using conditional $\alpha$-information, Liu et al.\cite{LCGR21} derived a universal bound on
the probability of success as follows.
\begin{theorem}[Generalized Fano's Inequality{\cite[Thm.~1]{LCGR21}}]
	\begin{equation}\label{fano}
		I_\alpha(K;\bm{Y}^m|T^m) \geq d_\alpha(\P_s(K|\bm{Y}^m,T^m) \| (\P_s(K)))
	\end{equation}
	where $d_\alpha(p\|q)$ is the binary $\alpha$-divergence:
	\begin{equation}
		d_\alpha(p\|q) =\tfrac{1}{\alpha-1} \log (p^\alpha q^{1-\alpha} +(1-p)^\alpha (1-q)^{1-\alpha}).
	\end{equation}
\end{theorem}
When \hbox{$\alpha \to 1$}, this bound recovers the previous bound in~\cite{CheriseyGuilleyRioulPiantanida19}. The simulation results in~\cite{LCGR21} show that~\eqref{fano} is tighter as $\alpha$ increases.

In this section, we prove that Fano's inequality for conditional $\alpha$-information becomes an \emph{equality} in the limiting case $\alpha = \infty$. Thus, conditional max-information can accurately characterize the probability of success.
\begin{theorem}[Generalized Fano's Inequality at $\alpha=+\infty$]\label{thm:1} {\color{blueR} For a uniformly distributed secret $K$,}
	\begin{equation}
\begin{split}
	\label{eq:1}
		I_\infty(K;\bm{Y}^m|T^m) &= d_\infty(\P_s(K|\bm{Y}^m,T^m) \| (\P_s(K))) \\
		&=\log (M\P_s)
\end{split}
	\end{equation}
	where $d_\infty(p \|q)=\lim\limits_{\alpha \to \infty} d_\alpha(p\|q)=\log \max\limits_{x,q(x)>0} (p(x)/q(x))$,
	$\P_s=\P_s(K|\bm{Y}^m,T^m)$ is the optimal probability of success, and $\P_s(K)=1/M$ is the corresponding probability of success in the case of blind estimation (without any observation).
\end{theorem}
To prove this theorem, we need the explicit expression of conditional max-information.
\begin{proposition}[Conditional Max-Information]
	Assuming $X$ takes values in a finite alphabet, one has
	\begin{equation}
		I_\infty(X;Y|Z) = \log \mathbb{E}_Z \int_y ( \mathop{\max}\limits_{x:p_{X|Z}(x|z)>0} p_{Y|XZ}) ~d\mu_Y \label{eq:I_infty}.
	\end{equation}
\end{proposition}
This result easily follows from the following Lemmas~\ref{lemma1} and~\ref{lemma2}, which are proved in Appendices~\ref{proof1} and~\ref{proof2} respectively. 
\textcolor{blueR}{In~\cite{Issa20}, conditional maximal leakage is defined as a maximum over $Z$, while our  conditional max-information is averaged over $Z$---which is less than or equal to the conditional maximal leakage of~\cite{Issa20}.}
\begin{lemma}\label{lemma1}
	Given any fixed $y,z$, we have
	\begin{equation}
		\lim_{\alpha \to \infty} ~p_{Y|Z} \cdot \langle p_{X|YZ}\|p_{X|Z}\rangle_\alpha = \max_{x:p_{X|Z}(x|z)>0} p_{Y|XZ}.
	\end{equation}
\end{lemma}

\begin{lemma}\label{lemma2}
	\begin{align}
		\lim_{\alpha \to \infty} \log & ~\E_{YZ} \langle p_{X|YZ}\|p_{X|Z}\rangle_\alpha \notag \\  &= \log \E_{Z} \int_y \lim_{\alpha \to \infty} ~p_{Y|Z} \cdot \langle p_{X|YZ}\|p_{X|Z}\rangle_\alpha.
	\end{align}
\end{lemma}

\begin{proof}[Proof of Theorem~\ref{thm:1}]
	Under the MAP rule, the probability of success writes
	\begin{align}
		\mathbb{P}_s &= \mathbb{E}_{\bm{Y}^mT^m} (\max_{k} ~p_{K|\bm{Y}^m,T^m}) \notag \\
		&= \mathbb{E}_{T^m} \int_{\bm{y}^m}(\max_{k} ~p_{\bm{Y}^m|K,T^m}p_{K|T^m}) d\mu_{\bm{Y}^m} \label{ps}.
	\end{align}
Recall $K$ is uniformly distributed and independent from $T^m$. Therefore, \eqref{ps} becomes
	\begin{equation}\label{P_s}
		\mathbb{P}_s= \frac{1}{M} \cdot \mathbb{E}_{T^m} \int_{\bm{y}^m} \bigl(\max_{k} ~p_{\bm{Y}^m|K,T^m} \bigr) d\mu_{\bm{Y}^m}.
	\end{equation}
Combining~\eqref{eq:I_infty} and~\eqref{P_s} we have $I_\infty(K;\bm{Y}^m|T^m) = \log(M \P_s)$.
	Since $\P_s\geq 1/M$, one has $\P_s\cdot M \geq (1-\P_s) \cdot M/(M-1)$ and
$d_\infty(\P_s(K|\bm{Y}^m,T^m) \| (\P_s(K))) =  \log (M\P_s)$,
	which proves~\eqref{eq:1}.
\end{proof}

\vspace{-.8em}
\subsection{Linear Bound Using \textcolor{blueR}{Maximal Leakage} $I_\infty(X;\mathbf{Y})$}

Evaluating $I_\infty(K;\bm{Y}^m|T^m)$ directly turns out to be cumbersome (see \hbox{Remark}~\ref{con-i-remark} below). Instead we use the unconditional max-information measure, \textcolor{blueR}{i.e., maximal leakage}  $I_\infty(X;\mathbf{Y})$  to bound the probability of success, which is linear in the number $m$ of measurements:
\begin{theorem}[Linear Bound]	\label{thm:linear-bound}
	\begin{equation}
		\log(M\P_s) \leq m I_\infty (X;\bm{Y}).
	\end{equation}
	\label{thm:linearBound}
\end{theorem}
\vspace{-1.5em}
\begin{proof}
	By Definition~\ref{maximal},
	\begin{equation}
		I_\infty (K,T^m;\bm{Y}^m) = \log \int_{\bm{y}^m} \max\limits_{k,t^m} ~p_{\bm{Y}^m|K,T^m} d\mu_{\bm{Y}^m}.
	\end{equation}
	Because $\max\limits_{k,t^m} ~p_{\bm{Y}^m|K,T^m} \geq
	\mathbb{E}_{T^m}~(\max_{k}~p_{\bm{Y}^m|K,T^m})
	$, by~\eqref{eq:1} and~\eqref{eq:I_infty} we have
	\begin{equation}
		I_\infty (K,T^m;\bm{Y}^m) \geq I_\infty (K;\bm{Y}^m|T^m) = \log (M\P_s).
	\end{equation}
	Because $(K,T^m) \leftrightarrow X^m \leftrightarrow {\bm{Y}}^m $ forms a Markov chain, using the {\color{blueR} data processing inequality (DPI)} for Sibson's $\alpha$-information~\cite{Polyanskiy2010,rioul2021primer} we have
	\begin{equation}\label{markov}
		I_\alpha (K,T^m;{\bm{Y}}^m ) \leq I_\alpha (X^m;{\bm{Y}}^m ).
	\end{equation}
	Also, when $T^m$ is not observed, each component of $X^m$ is i.i.d., and since the side-channel is memoryless, $(X^m;{\bm{Y}}^m)$ is an i.i.d. sequence. It easily follows from the definition that
	\begin{equation}\label{linear}
		I_\alpha (X^m;{\bm{Y}}^m) = m I_\alpha (X;\bm{Y}).
	\end{equation}
	Letting $\alpha \to \infty$ in~\eqref{markov} and~\eqref{linear} we have $I_\infty (K,T^m;\bm{Y}^m) \leq m I_\infty (X;\bm{Y})$.
\end{proof}

\begin{remark}\label{con-i-remark}
	For conditional $\alpha$-information we have the inequality  $I_\alpha(K;\bm{Y}^m|T^m) \leq I_\alpha(X^m;\bm{Y}^m|T^m) $ similar to~\eqref{markov}. However, one does not have an equality similar to~\eqref{linear} when $T^m$ is observed.
\end{remark}\vspace{-1em}
\textcolor{blueR}{\begin{remark}
	This proof cannot use the result in~\cite[Theorem~1]{Issa20} because in this theorem $\bm{Y^m}$ is not on a finite alphabet.
	What's more, if we use Definition 1 and Theorem 1 in~\cite{Issa20} we will have 
	\begin{equation}
		I_{\infty}(X^m;\bm{Y}^m, T^m) \geq \log (M\cdot \P_s(K|\bm{Y}^m,T^m))
	\end{equation}
	but $I_{\infty}(X^m;\bm{Y}^m)$ is less than $I_{\infty}(X^m;\bm{Y}^m, T^m)$.
\end{remark}}

\section{Mrs. Gerber's Lemma for Min-Entropy in Any Finite Abelian Group}\label{sec-MGL}

To benefit from Theorem~\ref{thm:linear-bound} it remains to upper bound $I_\infty (X;\bm{Y})$.
Since $X$ is uniformly distributed, it is easily seen from the definition that $I_\infty (X;\bm{Y})=\log M - H_\infty (X|\bm{Y})$. Thus, it remains to lower bound the conditional min-entropy $H_\infty (X|\bm{Y})$. This can be seen as an extension of Mrs. Gerber's lemma to min-entropy in finite additive groups.

\subsection{Mrs. Gerber's Lemma for Two Random Variables} 
Wyner and Ziv~\cite{Wyner1973ATO} lower bounded the entropy of a sum of binary random variables with the entropies of each summand. This is known as Mrs. Gerber's lemma.
\begin{theorem}[Mrs. Gerber's Lemma~\cite{Wyner1973ATO}] Let $X_0,X_1$ be two independent $\mathbb{Z}_2$-valued random variables with side information $\mathbf{Y}=(Y_0,Y_1)$ and sensitive bit $X=X_0\oplus X_1$. Then
	\begin{equation}
		H(X|\mathbf{Y}) \geq h( h^{-1}(H(\!X_0 |Y_0)) \star  h^{-1}(H(X_1 |Y_1)))
	\end{equation}
	where $h(p) = - p \log p - \bar{p} \log \bar{p}$, $a\star b= a \bar{b} + \bar{a} b$ and $\bar x =1-x$.
	\label{thm:originalMGL}
\end{theorem}
\vspace*{-1.5ex}

Jog and Anatharam~\cite{Jog2013TheEPI} extended Mrs. Gerber's lemma to additive groups of order $2^n$.
Hirche~\cite{Hirche2020} extended Mrs. Gerber's lemma for binary random variables to the case of Rényi entropies. In particular for min-entropy, one has equality:
\begin{theorem}[Christoph Hirche~{\cite[Lem.~IV.7]{Hirche2020}}] Let $X_0,X_1$ be two independent $\mathbb{Z}_2$-valued random variables with side information $\mathbf{Y} = (Y_0,Y_1)$ and $X=X_0\oplus X_1$. Then
	\begin{equation}
		H_\infty(X|\mathbf{Y}) = h_\infty( h_\infty^{-1}(H_\infty(X_0 |Y_0)) \star  h_\infty^{-1}(H_\infty(\!X_1 |Y_1)))
\vspace*{-1ex}
	\end{equation}
	where $h_\infty(p)= - \log \max \{p,\bar{p}\}$.
	\label{thm:inftyMGL}
\end{theorem}

In this section, Mrs. Gerber's lemma is extended for the min-entropy in any additive finite group:
\begin{theorem} 
	Let $X_0,X_1$ be two independent $\mathcal{G}$-valued random variables with side information $\mathbf{Y}=(Y_0,Y_1)$ and sensitive variable $X=X_0\oplus X_1$. 
	Then for $k = \max \{  \lfloor p^{-1} \!\rfloor , \lfloor q^{-1} \!\rfloor \}$, one has the optimal bound
	\begin{equation}
		\exp(\!-\!H_\infty\!(X|\mathbf{Y})) \!\leq\!
		\begin{cases}
			k p q +\!(1\!-\!kp)(1\!-\!kq) & \hspace{-1em} \text{ if } \!\frac{1}{k+1} \leq p,q \leq\! \frac{1}{k} \\
			\min \{ p,q\} & \hspace{-1em} \text{otherwise,}
		\end{cases}
		\label{eq-condcase}
	\end{equation}
	where $p = \exp(\!-\!H_\infty(X_0|Y_0))$ and $q = \exp(\!-\!H_\infty(X_1|Y_1))$.
	\label{thm:gerber-infty}
\end{theorem}
\begin{remark} Since
$k p q + (1\!-\!kp)(1\!-\!kq) \!=\! \tfrac{1}{k+1} \!+\! \tfrac{k}{k+1} (\!(k\!+\!1)p\!-\!1)(\!(k\!+\!1)q \!-\! 1)$,
$\frac{1}{k+1} \leq p,q \leq \frac{1}{k}$ implies $\frac{1}{k\!+\!1} \leq k p q \!+\! (1\!-\!kp)(1\!-\!kq) \leq \frac{1}{k}$. Thus, if both $H_\infty(X_0|Y_0)$ and $H_\infty(X_1|Y_1)$ lie in the interval $[\log k,\log(k+1)]$, then so does the corresponding bound on $H_\infty(X|\mathbf{Y})$.
\end{remark}

\begin{proof}
	We first prove the inequality in the unconditional case.
	The probability mass function of $X_0\oplus X_1$ is given by the convolution with respect to $\mathcal{G}$ of the probability mass functions of $X_0$ and $X_1$. That is, for any $x \in \mathcal{G}$,
	\begin{equation}
		 \P(X_0 \oplus X_1=x) = \sum_{i \in \mathcal{G}} \P(X_0 = x \oplus i) \P(X_1=\ominus i).
		 \vspace*{-2ex}
	\end{equation}
	In particular, \vspace*{-1ex}
	\begin{equation}
		\exp(\!-\!H_\infty(X_0 \oplus X_1)) = \max_{x \in \mathcal{G}} \sum_{i \in \mathcal{G}} \P(X_0 = x\oplus i) \P(X_1= \ominus i).
		\vspace*{-1.5ex}
	\end{equation}
	Hence the problem reduces to upper-bound		\vspace*{-.5ex}
	\begin{equation}
		\max_{x \in \mathcal{G}} \sum_{i \in \mathcal{G}} \P(X_0 = x \oplus i) \P(X_1= \ominus i).
		\vspace*{-0.5ex}
	\end{equation}
	Since 
	$\exp(\!-\!H_\infty(X_0 \ominus x)) = \exp(\!-\!H_\infty(X_0))$ we can assume without loss of generality that the maximum is reached for $x=0$ and the problem reduces to the maximization of
	\begin{equation}
		\sum_{i \in \mathcal{G}} \P(X_0 = i) \P(X_1= \ominus i).
	\end{equation}
	Let $(1),\ldots,(M) \in \mathcal{G}$ be an ordering of the group elements so that
		$\P(X_0 = (1)) \geq \P(X_0 = (2)) \geq \ldots \geq \P(X_0 = (M))$.
	The problem is to maximize		\vspace*{-1ex}
	\begin{equation}
		\sum_{i=1}^M \underbrace{\P(X_0 = (i))}_{p_{(i)}} \underbrace{\P(X_1 = \ominus (i))}_{q_{(i)}}.
		\vspace*{-1ex}
	\end{equation}
	The min-entropy of $X_1$ is invariant under any permutation of its probability mass function. Furthermore, by the \emph{rearrangement inequality} (Lemma~\ref{rearrangment-lemma} in Appendix~\ref{majorization}) a permutation of the probability mass function of $X_1$ maximizing the sum is such that
		$\P(X_1 \!=\! \ominus (1)) \geq \P(X_1 \!=\!  \ominus (2)) \geq \ldots \geq \P(X_1 \!=\! \ominus (\!M\!))$.
	Finally the problem is reduced to
	\begin{equation}
		\max_{\mathbf{p},\mathbf{q}} \phi(\mathbf{p},\mathbf{q}) {\color{blueR} \triangleq} \sum p_{(i)} q_{(i)}
		\label{eqn:max-two-pmfs}
		\vspace*{-1ex}
	\end{equation}
	under the constraint that $\exp(\!-\!H_\infty(X_0))= p_{(1)} = p$ and $\exp(\!-\!H_\infty(X_1)) = q_{(1)} = q$. Moreover, $h$ is Schur-convex in $\mathbf{p}$ when $\mathbf{q}$ is fixed and vice-versa (see Lemma~\ref{schur-combination}  in Appendix~\ref{majorization}). Hence the maximum in~\eqref{eqn:max-two-pmfs} is reached for the least spread out probability mass function under the min entropy constraints. That is  (Lemma~\ref{majorization-lemma} in Appendix~\ref{majorization}),
	\begin{equation}
		\begin{cases}
		(p_{(1)},\ldots,p_{(M)}) = (p,\ldots,p,1-k p,0,\ldots,0)
		\\
		(q_{(1)},\ldots,q_{(M)}\,) = (q,\ldots,q,1-l\,q,0,\ldots,0)
		\end{cases}
	\end{equation}
	where $k = \lfloor p^{-1}  \rfloor$ and $l = \lfloor q^{-1}  \rfloor$. Hence we obtain the bound
	\begin{equation}
		\exp(\!-\!H_\infty(X)) \leq
		\begin{cases}
			k p q + (1 - k p) (1 - k q) & \hspace{-1em} \text{ if } k=l
			\\
			\min \{ p, q \} & \hspace{-1em} \text{ otherwise.}
		\end{cases}
		\label{eq-uncondcase}
	\end{equation}
	
	It remains to prove that~\eqref{eq-uncondcase} carries over to the conditional case.
	Note that the bound is concave in $p$ for a fixed $q$ and vice-versa. Indeed, let $\frac{1}{k+1} \leq q \leq \frac{1}{k}$ be fixed. Then the inequality is piecewise linear in $p$, equal to
	\begin{equation}
		\begin{cases}
			p
			& \text{if } p \leq \frac{1}{k+1}
			\\
			k pq + (1-k p)  (1-k q)
			& \text{if } \frac{1}{k+1} \leq p \leq \frac{1}{k}
			\\
			q
			&\text{otherwise.}
		\end{cases}
	\end{equation}
	The three successive slopes are $1$, $k(k+1)q-k$ and $0$. Since $k(k+1)q-k \in [0,1]$, these slopes are in decreasing order and the function is indeed concave.
	Therefore, applying Jensen's inequality (twice) proves~\eqref{eq-condcase}. 
\end{proof}

\subsection{Extension to $d+1$ Summands}

Jog and Anatharam~\cite{Jog2013TheEPI} extended their generalization of Mrs. Gerber's lemma (for Shannon entropy) for random variables in group of order $2^n$ with two summands by repeating their inequality. In the same fashion, Theorem~\ref{thm:gerber-infty} is extended to $d+1$ summands by repeated application of Theorem~\ref{thm:gerber-infty}:

\begin{theorem} [Extension to $d+1$ summands]
	Let $p_i = \exp(\!-\!H_\infty(X_i|Y_i))$, without loss of generality  assume $p_0 \leq p_1 \leq \ldots \leq p_d$.  Let $k = \lfloor p_0^{-1} \rfloor$, $r = \max \{ i | p_i \leq \frac{1}{k}\}$. Then $H_d=H_\infty(X|\mathbf{Y})$ is lower bounded as
	\begin{equation}
		H_d \geq - \log \biggl(\frac{1}{k+1} + \frac{k^r}{k+1} \prod_{i=0}^r ((k+1)p_i-1) \biggl).
	\end{equation}
	\label{thm:extended-gerber-infty}
\end{theorem}
\vspace{-.8em}
\begin{proof} See Appendix~\ref{thm:proof-extended-gerber-infty}.
\end{proof}

In the side-channel context, it is particularly interesting to characterize the behavior of the inequality in the high entropy regime in terms of \textcolor{blueR}{maximal leakage}. This corresponds to the high noise regime of Theorem~\ref{thm:linear-bound}.
\begin{theorem}[Asymptotic for High Noise]
	Let $I_d = I_\infty(X;\mathbf{Y})$ in bits, then as $I_\infty(X_i;Y_i) \rightarrow 0$,
	\begin{equation}
		I_d \leq C_d \prod_{i=0}^d I_\infty(X_i;Y_i) + o\biggl( \prod_{i=0}^d I_\infty(X_i;Y_i) \biggl)
	\end{equation}
	where $C_d = (M-1)^d (\ln 2)^d$.
	\label{thm:taylor1}
\end{theorem}
\begin{proof}
	See Appendix~\ref{proof-taylor1}.
\end{proof}

\subsection{Refined Unconditioned Extension to $d+1$ Summands}

In contrast to Theorem~\ref{thm:gerber-infty}, Theorem \ref{thm:extended-gerber-infty} is not guaranteed to be optimal when $d>1$. The inequality can be improved by exploiting the structure of the sum of multiple random variables. We derive an improved bound which is optimal for entropies in the range $[\log(k\!-\!1),\log(k)]$ provided that there is a subgroup of $\mathcal{G}$ of order $k$. In particular, it is optimal in the high entropy regime $ [\log(M\!-\!1),\log(M)]$ (since the group itself is a subgroup of order $M$).

\begin{theorem}[Refined extension]
	Let $p_i = \exp(\!-\!H_\infty(X_i))$, without loss of generality we assume $p_0 \leq p_1 \leq \ldots \leq p_d$.  Let $k = \lfloor p_0^{-1} \!\rfloor$, $r = \max \{ i | p_i \leq \frac{1}{k}\}$. Let $H_d=H_\infty(X)$,
	\begin{equation}
		H_d \! \geq \!
		\begin{cases}
			\!-\! \log \bigl( \frac{1}{k+1} \!+\! \frac{1}{k+1} \prod \limits_{j=0}^r ((k\!+\!1)p_i \!-\! 1) \bigl) & \hspace{-1em} \text{ if $r$ is even,}
			\\
			\!-\! \log \bigl( \frac{1}{k+1} \!+\! \frac{k}{k+1} \prod \limits_{j=0}^r ((k\!+\!1)p_i \!-\! 1) \bigl) & \hspace{-1em} \text{ if $r$ is odd.}
		\end{cases}
	\end{equation}
	\label{thm:refined-extension}
\end{theorem}
\vspace{-1.5em}
\begin{proof} See Appendix~\ref{proof-refined}.
\end{proof}
Contrary to Theorem~\ref{thm:extended-gerber-infty}, Theorem~\ref{thm:refined-extension} does not apply to conditional min-entropy in general. In fact, when all the variables are fixed except one, the bound inside the logarithm is piece-wise linear but discontinuous in $\frac{1}{k}$ when $r$ is even. This discontinuity breaks the convexity of the inequality. Ensuring continuity for the desired convexity, we are led back to the expression of Theorem~\ref{thm:extended-gerber-infty}. However, under the assumption that
\begin{equation}
	\frac{1}{M} \leq \exp(\!-\!H_\infty(X_i|Y_i=y)) \leq \frac{1}{M-1}
 	\label{eq:hypothesis}
\end{equation}
for all $i$ and $y$,
 the bound of Theorem~\ref{thm:refined-extension} inside the logarithm is linear and we do obtain a conditional inequality. Fortunately, assumption~\eqref{eq:hypothesis} makes sense in the side-channel context. In fact, a common leakage model is $ Y_i = f_i(X_i) + \sigma \mathcal{N}(0,1)$ where $f_i$ is a fixed (possibly unknown) leakage function, such as the Hamming weight or a linear combination of the bits of the variable $X_i$. In particular \eqref{eq:hypothesis} holds for large enough $\sigma$ (high noise regime). Then we have the following

\begin{theorem}[Taylor Expansion] Assume \eqref{eq:hypothesis} and let $I_d = I_\infty(X;\mathbf{Y})$ in bits, then  as $I_\infty(X_i;Y_i) \rightarrow 0$,
	\begin{equation}
		I_d \leq C_d \prod_{j=0}^d I_\infty(X_i;Y_i)  + o\biggl( \prod_{j=0}^d I_\infty(X_i;Y_i) \biggl)
	\end{equation}
	where
	\vspace{-.5em}
	\begin{equation}
		C_d =
		\begin{cases}
			(\ln 2)^d & \text{ if $d$ is even,}
			\\
			(M-1) (\ln 2)^d & \text{ if $d$ is odd.}
		\end{cases}
	\label{equ:cd-pri}
	\end{equation}
	\label{thm:taylor2}
\end{theorem}
\vspace{-1.em}
\begin{proof} Taylor expansion of the exponential about $0$ and of the logarithm about $1$.
\end{proof}
Theorem~\ref{thm:taylor2} is particularly interesting because it suggests that, with respect to the \emph{worst} case leakage distribution, masking of \emph{odd} order $d$ is not useful compared to masking with order $d-1$ at high noise. In practice, however, for observed leakages this phenomenon may not apply. {\color{blueR} Theorem~\ref{thm:taylor2} is different from Theorem~\ref{thm:taylor1} as the constant $C_d$ is improved largely. Though Theorem~\ref{thm:taylor2} requires the high noise assumption \eqref{eq:hypothesis} to hold.}

Finally, combining Theorem~\ref{thm:taylor2} and Theorem~\ref{thm:linearBound} yields a bound on the probability of success
\begin{corollary}[Bound on $\P_s$]\label{cor-bound_PS} For $m$ traces, as $\P_s \rightarrow \tfrac{1}{M}$,
	\begin{equation}
		\P_s \leq \frac{\exp(m I_\infty(X;\mathbf{Y}))}{M} \approx \frac{1}{M} + \frac{m C_d}{M} \prod_{i=0}^d I_\infty(X_i;Y_i).
	\end{equation}
\end{corollary}
This is to be compared with the bound of \cite[Eqn.~8]{DBLP:journals/iacr/BeguinotCGLMRS22}:
\begin{proposition} 
	As $\P_s \rightarrow \tfrac{1}{M}$,\vspace*{-2ex}
	\begin{equation} 
		\P_s \leq \frac{1}{M} + \sqrt{m} A_d \biggl(\prod_{i=0}^d I(X_i,Y_i)\biggl)^{\tfrac{1}{2}}
	\end{equation}
	where
		$A_d = \sqrt{M-1} (2 \ln 2)^{\tfrac{d+1}{2}} M^{-1}$.
	\label{prop:cosade-eq}
\end{proposition} 
\vspace{-1.3em}
\begin{proof}
	See Appendix~\ref{proof:cosade-eq}.
\end{proof} 

As expected both bounds decrease exponentially in $d$ to the minimum value $\frac{1}{M}$. Although $I$ and $I_\infty$ are different metrics, we observe that
\begin{itemize}
	\item the constant factor $C_d/M$ for $I_\infty$ in~\eqref{equ:cd-pri} is exponentially lower in $d$ than the factor $A_d$ for $I$;
	\item the exponential decay in $d$ is twice higher for $I_\infty$;
	\item the inequality scales better for $I$ than for $I_\infty$ in terms of number $m$ of traces (since we compared both bounds for $\P_s \approx \frac{1}{M}$, $m$ is not necessarily taken large).
\end{itemize}

Finally, we can contrast both bounds on a toy example. Let $Y_i$ be uniformly distributed in $\{ x \in \mathcal{G} | x \neq X_i\}$. Then it is easily seen that $I(X_i,Y_i)=I_\infty(X_i,Y_i)=\log(\frac{M}{M-1})$. In this case,  the bound of this paper outperforms the bound of \cite{DBLP:journals/iacr/BeguinotCGLMRS22} in the high noise regime ($\P_s\to \frac{1}{M}$). Both bounds are compared numerically in Figs.~\ref{fig:d1} and~\ref{fig:d2} in Appendix~\ref{toy-example} for $d=1$ and $2$, respectively, and $M=256$.	

\section{Conclusion and Perspectives}\label{sec-concl}

We have shown that maximal leakage for masked implementations can be used to bound the probability of success of any side-channel attack. Maximal leakage is bounded by an efficiently computable bound based on a new variation of Mrs. Gerber's lemma for min-entropy. The bound tightness is commented with some example groups and probability mass function with figures in Appendix~\ref{discussion}. 
	
Improving the inequality when there is no subgroup of order $k+1$ in $\mathcal{G}$ is an interesting perspective. Indeed, groups of prime order which have no subgroup except the trivial ones are of major interest for their application to masking in asymmetric cryptographic schemes (especially post-quantum schemes). Besides, it would also be of interest to check whether the parity of $d$ does play a practical role in the efficiency of masked implementations.

\cleardoublepage

\IEEEtriggeratref{11}
\bibliographystyle{IEEEtranS}
\bibliography{ITbound_masking}

\begin{thebibliography}{10}
\providecommand{\url}[1]{#1}
\csname url@samestyle\endcsname
\providecommand{\newblock}{\relax}
\providecommand{\bibinfo}[2]{#2}
\providecommand{\BIBentrySTDinterwordspacing}{\spaceskip=0pt\relax}
\providecommand{\BIBentryALTinterwordstretchfactor}{4}
\providecommand{\BIBentryALTinterwordspacing}{\spaceskip=\fontdimen2\font plus
\BIBentryALTinterwordstretchfactor\fontdimen3\font minus
  \fontdimen4\font\relax}
\providecommand{\BIBforeignlanguage}[2]{{%
\expandafter\ifx\csname l@#1\endcsname\relax
\typeout{** WARNING: IEEEtranS.bst: No hyphenation pattern has been}%
\typeout{** loaded for the language `#1'. Using the pattern for}%
\typeout{** the default language instead.}%
\else
\language=\csname l@#1\endcsname
\fi
#2}}
\providecommand{\BIBdecl}{\relax}
\BIBdecl

\bibitem{DBLP:conf/ches/AkkarG01}
\BIBentryALTinterwordspacing
M.~Akkar and C.~Giraud, ``An implementation of {DES} and {AES}, secure against
  some attacks,'' in \emph{Cryptographic Hardware and Embedded Systems - {CHES}
  2001, Third International Workshop, Paris, France, May 14-16, 2001,
  Proceedings}, ser. Lecture Notes in Computer Science, {\c{C}}.~K. Ko{\c{c}},
  D.~Naccache, and C.~Paar, Eds., vol. 2162.\hskip 1em plus 0.5em minus
  0.4em\relax Springer, 2001, pp. 309--318. [Online]. Available:
  \url{https://doi.org/10.1007/3-540-44709-1\_26}
\BIBentrySTDinterwordspacing

\bibitem{Arimoto75}
S.~Arimoto, ``Information measures and capacity of order $\alpha$ for discrete
  memoryless channels,'' in \emph{Topics in Information Theory, Proc. 2nd
  Colloq. Math. Societatis J\'anos Bolyai}, A.~Joux, Ed., vol.~16, 1975, pp.
  41--52.

\bibitem{DBLP:journals/iacr/BeguinotCGLMRS22}
\BIBentryALTinterwordspacing
J.~B{\'e}guinot, W.~Cheng, S.~Guilley, Y.~Liu, L.~Masure, O.~Rioul, and F.-X.
  Standaert, ``Removing the field size loss from {D}uc et al.'s conjectured
  bound for masked encodings,'' \emph{{IACR} Cryptol. ePrint Arch.}, pp. 1--18,
  2022. [Online]. Available: \url{https://eprint.iacr.org/2022/1738}
\BIBentrySTDinterwordspacing

\bibitem{DBLP:conf/dsd/BeguinotCGR22}
\BIBentryALTinterwordspacing
J.~B{\'e}guinot, W.~Cheng, S.~Guilley, and O.~Rioul, ``Be my guess: {g}uessing
  entropy {vs.} success rate for evaluating side-channel attacks of secure
  chips,'' in \emph{25th Euromicro Conference on Digital System Design, {DSD}
  2022, Maspalomas, Spain, August 31 - Sept. 2, 2022}.\hskip 1em plus 0.5em
  minus 0.4em\relax {IEEE}, 2022, pp. 496--503. [Online]. Available:
  \url{https://doi.org/10.1109/DSD57027.2022.00072}
\BIBentrySTDinterwordspacing

\bibitem{cheng2022attacking}
W.~Cheng, Y.~Liu, S.~Guilley, and O.~Rioul, ``Attacking masked cryptographic
  implementations: {i}nformation-theoretic bounds,'' in \emph{2022 IEEE
  International Symposium on Information Theory (ISIT)}.\hskip 1em plus 0.5em
  minus 0.4em\relax IEEE, 2022, pp. 654--659.

\bibitem{CheriseyGuilleyRioulPiantanida19}
\BIBentryALTinterwordspacing
E.~{de Ch{\'e}risey}, S.~Guilley, O.~Rioul, and P.~Piantanida, ``Best
  information is most successful: {m}utual information and success rate in
  side-channel analysis,'' \emph{{IACR} Trans. Cryptogr. Hardw. Embed. Syst.},
  vol. 2019, pp. 49--79, 2019. [Online]. Available:
  \url{https://tches.iacr.org/index.php/TCHES/article/view/7385/6557}
\BIBentrySTDinterwordspacing

\bibitem{DucFS15}
\BIBentryALTinterwordspacing
A.~Duc, S.~Faust, and F.-X. Standaert, ``Making masking security proofs
  concrete - or how to evaluate the security of any leaking device,'' in
  \emph{Advances in Cryptology - {EUROCRYPT} 2015 - 34th Annual International
  Conference on the Theory and Applications of Cryptographic Techniques, Sofia,
  Bulgaria, April 26-30, 2015, Proceedings, Part {I}}, ser. Lecture Notes in
  Computer Science, E.~Oswald and M.~Fischlin, Eds., vol. 9056.\hskip 1em plus
  0.5em minus 0.4em\relax Springer, 2015, pp. 401--429. [Online]. Available:
  \url{https://doi.org/10.1007/978-3-662-46800-5\_16}
\BIBentrySTDinterwordspacing

\bibitem{FehrBerens14}
S.~Fehr and S.~Berens, ``On the conditional {R}{\'e}nyi entropy,'' \emph{IEEE
  Transactions on Information Theory}, vol.~60, pp. 6801--6810, 2014.

\bibitem{DBLP:conf/ches/GoubinP99}
\BIBentryALTinterwordspacing
L.~Goubin and J.~Patarin, ``{DES} and differential power analysis (the
  ``duplication'' method),'' in \emph{Cryptographic Hardware and Embedded
  Systems, First International Workshop, CHES'99, Worcester, MA, USA, August
  12-13, 1999, Proceedings}, ser. Lecture Notes in Computer Science,
  {\c{C}}.~K. Ko{\c{c}} and C.~Paar, Eds., vol. 1717.\hskip 1em plus 0.5em
  minus 0.4em\relax Springer, 1999, pp. 158--172. [Online]. Available:
  \url{https://doi.org/10.1007/3-540-48059-5\_15}
\BIBentrySTDinterwordspacing

\bibitem{Hirche2020}
C.~Hirche, ``{R}{\'e}nyi bounds on information combining,'' in \emph{2020 IEEE
  International Symposium on Information Theory (ISIT)}, 2020, pp. 2297--2302.

\bibitem{HoVerdu15}
S.-W. Ho and S.~Verd{\'u}, ``Convexity/concavity of {R}{\'e}nyi entropy and
  $\alpha$-mutual information,'' in \emph{2015 IEEE International Symposium on
  Information Theory (ISIT)}, 2015, pp. 745--749.

\bibitem{Issa20}
I.~Issa, A.~B. Wagner, and S.~Kamath, ``An operational approach to information
  leakage,'' \emph{IEEE Transactions on Information Theory}, vol.~66, no.~3,
  pp. 1625--1657, 2020.

\bibitem{DBLP:conf/ccs/ItoUH22}
\BIBentryALTinterwordspacing
A.~Ito, R.~Ueno, and N.~Homma, ``On the success rate of side-channel attacks on
  masked implementations: {i}nformation-theoretical bounds and their practical
  usage,'' in \emph{Proceedings of the 2022 {ACM} {SIGSAC} Conference on
  Computer and Communications Security, {CCS} 2022, Los Angeles, CA, USA,
  November 7-11, 2022}, H.~Yin, A.~Stavrou, C.~Cremers, and E.~Shi, Eds.\hskip
  1em plus 0.5em minus 0.4em\relax {ACM}, 2022, pp. 1521--1535. [Online].
  Available: \url{https://doi.org/10.1145/3548606.3560579}
\BIBentrySTDinterwordspacing

\bibitem{Jog2013TheEPI}
V.~Jog and V.~Anantharam, ``The entropy power inequality and {M}rs. {G}erber's
  lemma for groups of order $2^n$,'' \emph{2013 IEEE International Symposium on
  Information Theory (ISIT)}, pp. 594--598, 2013.

\bibitem{LCGR21}
\BIBentryALTinterwordspacing
Y.~Liu, W.~Cheng, S.~Guilley, and O.~Rioul, ``On conditional alpha-information
  and its application to side-channel analysis,'' in \emph{{IEEE} Information
  Theory Workshop, {ITW} 2021, Kanazawa, Japan, October 17-21, 2021}.\hskip 1em
  plus 0.5em minus 0.4em\relax {IEEE}, 2021, pp. 1--6. [Online]. Available:
  \url{https://doi.org/10.1109/ITW48936.2021.9611409}
\BIBentrySTDinterwordspacing

\bibitem{madiman2021entropy}
M.~Madiman, L.~Wang, and J.~O. Woo, ``Entropy inequalities for sums in prime
  cyclic groups,'' \emph{SIAM Journal on Discrete Mathematics}, vol.~35, no.~3,
  pp. 1628--1649, 2021.

\bibitem{Marshall1980InequalitiesTO}
A.~W. Marshall, I.~Olkin, and B.~C. Arnold, \emph{Inequalities: {T}heory of
  Majorization and Its Applications}.\hskip 1em plus 0.5em minus 0.4em\relax
  Springer, 1980.

\bibitem{DBLP:journals/iacr/MasureRS22}
\BIBentryALTinterwordspacing
L.~Masure, O.~Rioul, and F.-X. Standaert, ``A nearly tight proof of {D}uc et
  al.'s conjectured security bound for masked implementations,'' \emph{{IACR}
  Cryptol. ePrint Arch.}, p. 600, 2022. [Online]. Available:
  \url{https://eprint.iacr.org/2022/600}
\BIBentrySTDinterwordspacing

\bibitem{Polyanskiy2010}
Y.~Polyanskiy and S.~Verd{\'u}, ``Arimoto channel coding converse and
  {R}{\'e}nyi divergence,'' in \emph{2010 48th Annual Allerton Conference on
  Communication, Control, and Computing (Allerton)}, 2010, pp. 1327--1333.

\bibitem{rioul2021primer}
O.~Rioul, ``A primer on alpha-information theory with application to leakage in
  secrecy systems,'' in \emph{International Conference on Geometric Science of
  Information}.\hskip 1em plus 0.5em minus 0.4em\relax Springer, 2021, pp.
  459--467.

\bibitem{ShamaiWyner90}
S.~Shamai and A.~Wyner, ``A binary analog to the entropy-power inequality,''
  \emph{IEEE Transactions on Information Theory}, vol.~36, no.~6, pp.
  1428--1430, 1990.

\bibitem{Sibson69}
R.~Sibson, ``Information radius,'' \emph{Zeitschrift f{\"u}r
  Wahrscheinlichkeitstheorie und verwandte Gebiete}, vol.~14, pp. 149--160,
  1969.

\bibitem{Tao2009SumsetAI}
T.~Tao, ``Sumset and inverse sumset theory for {S}hannon entropy,''
  \emph{Combinatorics, Probability and Computing}, vol.~19, pp. 603 -- 639,
  2009.

\bibitem{teixeira2012}
A.~Teixeira, A.~Matos, and L.~Antunes, ``Conditional {R}{\'e}nyi entropies,''
  \emph{IEEE Transactions on Information Theory}, vol.~58, no.~7, pp.
  4273--4277, 2012.

\bibitem{ErvenHarremos14}
T.~van Erven and P.~Harremos, ``R{\'e}nyi divergence and kullback-leibler
  divergence,'' \emph{IEEE Transactions on Information Theory}, vol.~60, no.~7,
  pp. 3797--3820, 2014.

\bibitem{Verdu15}
\BIBentryALTinterwordspacing
S.~Verd{\'u}, ``$\alpha$-mutual information,'' in \emph{IEEE Information Theory
  and Applications Workshop (ITA2015)}, San Diego, USA, 2015, pp. 1--6.
  [Online]. Available: \url{https://doi.org/10.1109/ITA.2015.7308959}
\BIBentrySTDinterwordspacing

\bibitem{Wyner1973ATO}
A.~D. Wyner and J.~Ziv, ``A theorem on the entropy of certain binary sequences
  and applications-{I},'' \emph{IEEE Transactions on Information Theory},
  vol.~19, pp. 769--772, 1973.

\end{thebibliography}

\cleardoublepage


\begin{appendix}

	\subsection{Background on Majorization}\label{majorization}
	
	We recall definitions and basic results of majorization theory. An extensive presentation can be found in the reference textbook \cite{Marshall1980InequalitiesTO}.
	
	\begin{definition}[Statistical Ordering] If $\mathbf{p} = (p_1,\ldots,p_M)$ is a probability mass function, an arrangement $(1),(2),\ldots,(M)$ of
		$\mathbf{p}$ so that $p_{(1)} \geq \ldots \geq p_{(M)}$ is said to be the statistical ordering of $\mathbf{p}$.
		The associated cumulative mass function is noted $P_{(i)} = p_{(1)} + \ldots + p_{(i)}$ where $P_{(0)}=0$ by convention.
	\end{definition}
	
	\begin{definition}[Majorization] Let $\mathbf{p},\mathbf{q}$ be two probability mass functions. We say that $\mathbf{q}$ majorizes $\mathbf{p}$
		and write $\mathbf{p} \preceq \mathbf{q}$ if
		\begin{equation}
			P_{(i)} \leq Q_{(i)} \qquad (i=1,\ldots,M).
		\end{equation}
		This partial order on the probability mass functions quantifies whether a distribution is more spread out than the other.
	\end{definition}
	
	\begin{definition}[Schur-Convexity] $f : \mathbf{p} \mapsto f(\mathbf{p}) \in \mathbb{R}$ is said to be Schur-convex if it is increasing with respect to majorization i.e. $\mathbf{p} \preceq \mathbf{q} \implies f(\mathbf{p}) \leq f(\mathbf{q})$.
	\end{definition}

	\begin{lemma}[Schur-Convex Combination] If $\alpha_1 \geq \ldots \geq \alpha_M$ then $(p_1,\ldots,p_M) \mapsto \sum_{i=1}^M \alpha_i p_{(i)}$ is Schur-convex.
		\label{schur-combination}
	\end{lemma}
	
	\begin{proof} This can be shown by an Abel transform as pointed out in~\cite[Remark 2]{DBLP:conf/dsd/BeguinotCGR22}.
		\begin{align}
			\sum \alpha_i p_{(i)}
			&= \sum \alpha_i (P_{(i)} - P_{(i-1)})
			\\
			&= \alpha_M P_{(M)} - \alpha_1 P_{(0)} - \sum_{i=1}^{M-1} (\alpha_{i+1}-\alpha_i) P_{(i)}
			\\
			&= \alpha_M - \sum_{i=1}^{M-1} (\alpha_{i+1}-\alpha_i) P_{(i)}.
		\end{align}
		Since $\alpha_{i+1}-\alpha_i \leq 0$ the Schur-convexity follows from the definition.
	\end{proof}
	
	\begin{lemma}[Majorization and min-entropy] Let $\mathbf{p}$ be a probability mass functions whose min-entropy is equal to $-\log p$ and $k = \lfloor p^{-1} \!\rfloor$ then
		\begin{equation}
			(p,\frac{1\!-\!p}{M\!-\!1},\ldots,\frac{1\!-\!p}{M\!-\!1}) \preceq \mathbf{p} \preceq (p,\ldots,p,1\!-\!kp,0,\ldots,0).
		\end{equation}
		\label{majorization-lemma}
	\end{lemma}
	
	\begin{lemma}[Rearrangement Inequality] Let $(a_1,\ldots,a_n)$, $(b_1,\ldots,b_n) \in \mathbb{R}^{+n}$ be two sequences in descending order. Then for all permutations
		$\sigma$ of $\{1,\ldots,n\}$ it holds that
		\begin{equation}
			\sum a_i b_{n+1-i} \leq \sum a_i b_{\sigma(i)} \leq \sum a_i b_i.
		\end{equation}
		\label{rearrangment-lemma}
	\end{lemma}
	\begin{proof}
		See \cite{Marshall1980InequalitiesTO} for a proof using majorization.
	\end{proof}

	\subsection{Proof of Lemma~\ref{lemma1}}\label{proof1}
	\textcolor{blueR}{
	\textbf{Method 1}:\\
	By Theorem 6 of~\cite{ErvenHarremos14} we have 
	\begin{align}
		\lim\limits_{\alpha \to \infty}  \langle p_{X|YZ}\|p_{X|Z}\rangle_\alpha = \exp \big( D_{\infty} (P_{X|YZ}\|P_{X|Z}) \big) \\
		= \max_{x:p_{X|Z}(x|z)>0} \frac{p_{X|YZ}}{p_{X|Z}}.
	\end{align}
	Because $p_{Y|Z}\cdot p_{X|YZ}/p_{X|Z} =p_{Y|XZ}$, the proof is finished.\\
	\textbf{Method 2}:} \\
	We use $L^{\infty}$-norm to prove this lemma.
	\begin{align}
		p_{Y|Z} & \langle  p_{X|YZ} \|p_{X|Z}\rangle_\alpha = p_{Y|Z} ~\bigl( \sum_{x \in \mathcal{X}} p^{\alpha}_{X|YZ} ~p^{1-\alpha}_{X|Z} \bigr) ^{\frac{1}{\alpha}} \notag \\
		&= \bigl( \sum_{x \in \mathcal{X}} p^{\alpha}_{XY|Z}~ p^{1-\alpha}_{X|Z} \bigr) ^{\frac{1}{\alpha}} = \Bigl( \sum_{x \in \mathcal{X}} \bigl( p_{XY|Z}~ p^{\frac{1-\alpha}{\alpha}}_{X|Z} \bigr)^{\alpha} \Bigr) ^{\frac{1}{\alpha}} \notag \\
		&=  \Bigl( \sum_{x \in \mathcal{X}} \bigl( p_{Y|XZ}~ p^{\frac{1}{\alpha}}_{X|Z} \bigr)^{\alpha} \Bigr) ^{\frac{1}{\alpha}}. \label{expression}
	\end{align}
	For any $\varepsilon >0$, there exists a sufficiently large $\alpha >0$ such that
	\begin{equation} \label{ineq}
		p_{Y|XZ} -\varepsilon \leq p_{Y|XZ}~ p^{\frac{1}{\alpha}}_{X|Z} \leq p_{Y|XZ}.
	\end{equation}
	Because $\mathcal{X}$ is finite, one always has a sufficiently large $\alpha >0$ such that \eqref{ineq} holds for any $x \in \mathcal{X}$.
	By $L^{\infty}$-norm we have
	\begin{equation}
		\lim_{\alpha \to \infty}   \Bigl(\! \sum_{x:p_{X|Z}(x|z)>0} \hspace{-1em} \bigl( p_{Y|XZ} -\varepsilon  \bigr)^{\alpha} \Bigr) ^{\frac{1}{\alpha}} = \hspace{-1em} \max_{x:p_{X|Z}(x|z)>0} p_{Y|XZ} -\varepsilon
	\end{equation}
	Since $\varepsilon >0$ is arbitrary, combined with the squeeze theorem, the proof is finished.
	
	\subsection{Proof of Lemma~\ref{lemma2}}\label{proof2}
	\textcolor{blueR}{
By definition we have
\begin{align}
	& \lim\limits_{\alpha \to \infty} \log \mathbb{E}_{YZ} \langle p_{X|YZ}\|p_{X|Z}\rangle_\alpha  \notag\\
	&= \lim\limits_{\alpha \to \infty} \log \mathbb{E}_{YZ} \exp \bigl(\tfrac{\alpha-1}{\alpha} D_{\alpha} \langle p_{X|YZ}\|p_{X|Z}\rangle_\alpha\bigr) . 
\end{align}
This value is bounded because $I(X;Y|Z) \leq \log M$. Since $\tfrac{\alpha-1}{\alpha}D_{\alpha} \langle p_{X|YZ}\|p_{X|Z}\rangle_\alpha$ is increasing in $\alpha$, the lemma follows from the monotone convergence theorem.
	 \hfill\IEEEQED  }
	
	\subsection{Proof of Theorem~\ref{thm:extended-gerber-infty}}\label{thm:proof-extended-gerber-infty}
	
	We prove the inequality by induction. Theorem~\ref{thm:gerber-infty} settles the case of $d+1=2$ variables. We assume it is true for all sets of at most $d+1$ variables and show it is true for all set of at most $d+2$ variables.  Let $k,r_{d+1}$ be the value of $k,r$ in the theorem associated to $X_0,\ldots,X_{d},X_{d+1}$. If $r_{d+1} < d+1$ we lower bound the min-entropy of the sum $X_0 \oplus \ldots \oplus X_{d+1}$ by the entropy of
	$X_0 \oplus \ldots \oplus X_{d}$. We conclude by applying the induction hypothesis to this sum of $d$ random variables. Else
	$r_{d+1} = d+1$. Since $X_0 \oplus \ldots \oplus X_d \oplus X_{d+1} = (X_0 \oplus \ldots \oplus X_{d}) \oplus X_{d+1}$, we apply the induction hypothesis {\color{blueR} $\mathcal{H}_d$ } to $X_0,\ldots,X_d$ then we apply Theorem.~\ref{thm:gerber-infty} to $X_{d+1}$ and $X_{0} \oplus \ldots \oplus X_d$. Let $K = \exp(\!-\!H_\infty(X_0 \oplus\ldots\oplus X_{d+1}|Y_0\ldots Y_{d+1}))$.
	\begin{align}
		K
		&\overset{(a)}{\leq} 1 - k (p_{d+1} + \exp(\!-\!H_\infty(X_0\!\oplus\!\ldots\!\oplus\!X_d|Y_0\!\ldots\!Y_d))) \nonumber
		\\
		&+ k(k+1) p_{d+1} \exp(\!-\!H_\infty(X_0\!\oplus\!\ldots\!\oplus\!X_d|Y_0\!\ldots\!Y_d))
		\\
		&\overset{(b)}{\leq} 1 - k( p_{d+1} + \frac{1}{k+1} + \frac{k^d}{k+1} \prod_{i=0}^d ((k+1)p_i-1)) \nonumber
		\\
		&+ k(k+1)p_{d+1} (\frac{1}{k+1} + \frac{k^d}{k+1} \prod_{i=0}^d ((k+1)p_i-1))
		\\
		&= \frac{1}{k+1} + \frac{k^{d+1}}{k+1} \prod_{i=0}^d ((k+1)p_i-1)) ((k+1)p_{d+1}-1)
		\\
		&= \frac{1}{k+1} + \frac{k^{d+1}}{k+1} \prod_{i=0}^{d+1} ((k+1)p_i-1))
	\end{align} {\color{blueR} where $(a)$ holds by $\mathcal{H}_1$ and $(b)$ holds by $\mathcal{H}_d$.}
	As a repeated application of Theorem~\ref{thm:gerber-infty} the inequality naturally extends to the conditional case.\hfill\IEEEQED

	\subsection{Proof of Theorem~\ref{thm:taylor1}}\label{proof-taylor1}
	
	We upper bound $I_d= \log M - H_\infty(X|\mathbf{Y})$ using the lower bound on
	the min entropy. At high entropy $k=M-1$ hence
	\begin{equation}
		\log M - I_d \geq -\log\biggl(\frac{1}{M} + \frac{(M-1)^d}{M} \prod_{i=0}^d (Mp_i-1)\biggl)
	\end{equation}
	where
	\begin{equation}
		p_i = \frac{\exp(I_\infty(X_i;Y_i))}{M}.
	\end{equation}
	
	\begin{align}
		I_d
		&\leq  \log\biggl( 1 + (M-1)^d \prod_{i=0}^d \bigl(\exp(I_\infty(X_i;Y_i)) -1\bigl)\biggl)
		\\
		&= (\!M\!-\!1\!)^d \! (\ln 2)^d \! \prod_{i=0}^d I_\infty(X_i;Y_i) \!+\! o\biggl( \prod_{i=0}^d I_\infty(\!X_i;Y_i\!) \!\biggl)
	\end{align}
	\hfill\IEEEQED
	
	\subsection{Proof of Theorem~\ref{thm:refined-extension}}\label{proof-refined}
	
	{\color{blueR} We first prove the following usefull lemma. It} intuitively tells that to minimize the min-entropy the pmf should not spread out to other values. For instance when the summed random variables are in a sub-group the value of their sum is confined in this sub-group.
	\begin{lemma} If $X_0$ and $X_1$ have pmfs up to permutation $(q,\ldots,q,1-kq,0,\ldots,0)$ and $(p,\ldots,p,1-kp,0,\ldots,0)$
		then the pmf of $X_0 \oplus X_1$ is majorized by the pmf $(r,\ldots,r,1-kr,0,\ldots,0)$ where $r=p+q-(k+1)pq$. There is equality when $X_0,X_1$ are supported on the coset of a subgroup of $\mathcal{G}$ of order $k+1$.
		\label{convolution-lemma}
	\end{lemma}
	
	\begin{proof} The convolution involves $(k+1)^2$ strictly positive terms. Namely
		\begin{equation}
			\begin{cases}
				pq & k^2 \text{ times}
				\\
				p(1-kq) & k \text{ times }
				\\
				q(1-kp) & k \text{ times }
				\\
				(1-kp)(1-kq) & \text{ once }
			\end{cases}.
		\end{equation}
		Further, in each mass of the results they are at most $k+1$ terms that are added and at most once an expression containing $(1-kp)$ and $(1-kq)$. Let us assume that $q \geq 1-kq$ and $p \geq 1-kp$ or $1-kq \geq q$ and $1-kp\geq p$. By rearrangement inequality (Lemma~\ref{rearrangment-lemma}), $kpq + (1-kp)(1-kq)$ is the largest terms that can be obtained. The $2^{\text{nd}}$ to $(k\!+\!1)$-th largest terms are $(k\!-\!2)pq + p(1\!-\!kq)+q(1\!-\!kp)=p\!+\!q\!-\!(k\!+\!1)pq$. This majorizes all possible results since each term of the statistical ordering is maximized the sequence of cumulative mass function is also maximized. If $q \geq 1-kq$ and $1-kp \geq p$ or $1-kq\geq q$ and $p \geq 1-kp$ the proof is the same but the $1^{\text{st}}$ to $k$-th largest terms are $(k\!-\!2)pq + p(1\!-\!kq)+q(1\!-\!kp)=p\!+\!q\!-\!(k\!+\!1)pq$ and the $k+1$-th largest term is $kpq + (1-kp)(1-kq)$ which is the same pmf up to a permutation. The case of equality is clear.
	\end{proof}
	
	We derive the inequality of Thm.~\ref{thm:refined-extension}.
	{\color{blueR} The proof is composed of three steps. The first step is to prove that the inequality is achieved for 
	pmf of the form $\mathbf{p_j} = (p_{j},\ldots,p_{j},1-k_j p_{j},0,\ldots,0)$. The second step is to majorize the 
	resulting convolution by induction. The final steps is to conclude the majorization argument.
	}
	As in the case of two summands the problem is to maximize
	\begin{equation}
		\max_{x \in \mathcal{G}} \sum_{i_0,i_1,\ldots,i_{d-1} \in \mathcal{G}} \biggl(\prod_{j=0}^{d-1} \P(X_j = i_j)\biggl) \P(X_d=x\ominus\bigoplus_{j=0}^{d-1} i_j).
	\end{equation}
	
	Without loss of generality, we can assume that the maximum is reached in $x=0$, it remains to upper bound
	\begin{equation}
		\phi(\mathbf{p_0},\ldots,\mathbf{p_d}) \triangleq \hspace{-2em} \sum_{i_0,i_1,\ldots,i_{d-1}  \in \mathcal{G}} \biggl(\prod_{j=0}^{d-1} \P(X_j \!=\! i_j)\!\biggl) \P(X_d=\ominus\bigoplus_{j=0}^{d-1} i_j).
	\end{equation}
	
	We fix $\mathbf{p_1},\ldots,\mathbf{p_d}$. The maximization can be written as
	\begin{equation}
		\sum_{i_0 \in \mathcal{G}} \P(X_0=i_0) \alpha_{i_0}
	\end{equation}
	where
	\begin{equation}
		\alpha_{i_0} = \hspace{-0.6em}\sum_{i_1,\ldots,i_{d-1} \in \mathcal{G}} \biggl(\prod_{j=1}^{d-1} \P(X_j = i_j)\biggl) \P(X_d= \ominus\bigoplus_{j=0}^{d-1} i_j).
	\end{equation}
	
	This is equivalent to maximize
	\begin{equation}
		\sum_{i=1}^M \P(X_0=(i)) \alpha_{(i)}
		\label{eqn:alpha-i-sum}
	\end{equation}
	where $(1),\ldots,(M)$ are such that $\P(X_0=(1)) \geq \ldots \geq \P(X_0=(M))$. By rearrangement (Lemma~\ref{rearrangment-lemma}), \eqref{eqn:alpha-i-sum} is maximum when
	$ \alpha_{(1)} \geq \ldots \geq \alpha_{(M)}.$ By lemma~\ref{schur-combination} this mapping is Schur-Convex in $\mathbf{p_0}$ hence by lemma~\ref{majorization-lemma} it is maximized for statistical ordering of the probability mass function of $X_0$ of the form
	\begin{equation}
		\mathbf{p_0} = (p_0,\ldots,p_0,1-k_0p_0,0,\ldots,0)
		\label{eqn:shape}
	\end{equation}
	where $k_0 = \lfloor p_0^{-1} \!\rfloor$. Equation~\eqref{eqn:shape} does not depend on the fixed probability mass functions of $X_1,\ldots,X_d$. By symmetry, we also obtain that the for $j=0,\ldots,d$ the statistical ordering of the probability mass function of $X_j$ is of the form
	\begin{equation}
		\mathbf{p_j} = (p_j,\ldots,p_j,1-k_jp_j,0,\ldots,0)
	\end{equation}
	where $k_j = \lfloor p_j^{-1} \!\rfloor$. {\color{blueR} This concludes the first step of the proof. As previous proof we can further assume without loss of generality that $k_j$ is constant equal to $k$ for all $j$.} It remains to determine for which permutation of these probability mass function we obtain the lowest min-entropy.
		
	Now we fix the pmf $\mathbf{p_2},\ldots,\mathbf{p_d}$. And we consider the maximization with respect to the pmf of $X_0+X_1$. By lemma~\ref{schur-combination}, the expression is Schur-convex. Hence it is maximized for the least spread out pmf. By lemma \ref{convolution-lemma}, the pmf is majorized by $(r,\ldots,r,1-kr,0,\ldots,0)$ where \begin{equation}
		r=p+q-(k+1)pq.
		\label{eqn:r-rec}
	\end{equation}
	
	We can proceed by induction to majorize the sum of $d+1$ random variables. {\color{blueR}  Let $\mathcal{H}_d$ be the induction hypothesis:
	The probability mass function of the sum of $d+1$ random variables is majorized by
	$(r,\ldots,r,1-kr,0,\ldots,0)$ where
	\begin{equation}
		(k+1)r = 1 + (-1)^d \prod_{i=0}^d ((k+1)p_i-1).
	\end{equation} 
	The initialization $\mathcal{H}_1$ is true from \eqref{eqn:r-rec}. We assume $\mathcal{H}_j$ holds and proves $\mathcal{H}_{j+1}$ holds.
	Using \eqref{eqn:r-rec} with $\mathcal{H}_j$ we obtain that the convolution is majorized by $(r,\ldots,r,1-kr,0,\ldots,0)$ with 
	\begin{align}
		r
		&= p_{j+1} + \frac{1}{k+1} + \frac{(-1)^j}{k+1} \prod_{i=0}^j ((k+1)p_i-1) 
		\\
		&- (k+1)p_{j+1} \left( \frac{1}{k+1} + \frac{(-1)^j}{k+1} \prod_{i=0}^j ((k+1)p_i-1) \right)
		\\
		&= \frac{1}{k+1} + \frac{(-1)^j}{k+1} \prod_{i=0}^j ((k+1)p_i-1) (1 - (k+1)p_{j+1})
	\end{align}
	This proves $\mathcal{H}_{j+1}$ and we conclude by induction. This concludes the second step of the proof and it remains to conclude.
	}

	We proved that the probability mass function of the sum of $d+1$ random variables is majorized by
	$(r,\ldots,r,1-kr,0,\ldots,0)$ where
	\begin{equation}
		(k+1)r = 1 + (-1)^d \prod_{i=0}^d ((k+1)p_i-1).
		\label{eqn:r-refined}
	\end{equation}
	
	This shows that
	\begin{align*}
		\exp( \!-\! H_d )
		&\leq \!
		\begin{cases}
			r & \text{ if } $d$ \text{ is even}
			\\
			1 - k r & \text{ if } $d$ \text{ is odd}
		\end{cases}
		\\
		&\!= \!
		\begin{cases}
			\frac{1}{k+1} \!+\! \frac{1}{k+1} \prod_{j=0}^d ((k\!+\!1)p_i\!-\!1) & \hspace{-1em} \text{ ($d$ even) }
			\\
			\frac{1}{k+1} \!+\! \frac{k}{k+1} \prod_{j=0}^d ((k\!+\!1)p_i\!-\!1) & \hspace{-1em} \text{ ($d$ odd) }
		\end{cases}.
	\end{align*}
	\vspace*{-2ex}\hfill\IEEEQED
	
	\subsection{Proof of Proposition~\ref{prop:cosade-eq}}\label{proof:cosade-eq}

	Using Fano's inequality, de Chérisey et al. \cite[Eqn.~11]{CheriseyGuilleyRioulPiantanida19} have shown that
	\begin{equation}
		m I(X;\mathbf{Y}) \geq \log(M) - h(\P_s) - (1-\P_s)\log(M-1).
		\label{eq:eloi-linear}
	\end{equation}
	This can be explicited by computing a Taylor expansion of degree two of the binary entropy function in $\P_s=\tfrac{1}{M}$,
	\begin{align}
		h(\P_s)
		&= h(\tfrac{1}{M}) + h'(\tfrac{1}{M}) (\P_s-\tfrac{1}{M})
		\nonumber
		\\
		&+\frac{h''(\tfrac{1}{M})}{2} (\P_s-\tfrac{1}{M})^2 + o( (\P_s-\tfrac{1}{M})^2)
		\\
		&= \log(M) \!-\! (1\!-\!\frac{1}{M}) \log(M\!-\!1) \!+\! \log(M \!-\! 1) (\P_s \!-\! \tfrac{1}{M})
		\nonumber 
		\\
		&- \frac{M^2 \log(e)}{2(M-1)} (\P_s-\tfrac{1}{M})^2 + o( (\P_s-\tfrac{1}{M})^2).
	\end{align}
	In particular \eqref{eq:eloi-linear} reduces to 
	\begin{equation}
		m I(X;\mathbf{Y}) \geq \frac{M^2 \log(e)}{M-1} (\P_s-\tfrac{1}{M})^2 + o( (\P_s-\tfrac{1}{M})^2),
		\label{eq:lin-lin}
	\end{equation}
	\noindent
	where we leveraged the following equalities 
	\begin{equation}
		h(\tfrac{1}{M}) = \log(M) - (1-\frac{1}{M}) \log(M-1),
	\end{equation}
	\begin{equation}
		h'(\tfrac{1}{M})=\log(M-1) \text{ and } h''(\tfrac{1}{M})= \frac{- M^2 \log(e)}{M-1}.
	\end{equation}
	In particular,~\eqref{eq:lin-lin} shows that, 
	\begin{align}
		\P_s
		&\leq \frac{1}{M} + \sqrt{ \frac{2 \ln 2(M-1) m }{M^2} I(X,\mathbf{Y}) }
		\\
		&\approx \frac{1}{M} + \sqrt{m} A_d \sqrt{ \prod_{i=0}^d I(X_i,Y_i)  } & \text{ (with \cite[Eqn.~8]{DBLP:journals/iacr/BeguinotCGLMRS22})}
	\end{align}
	where
	\begin{equation}
		A_d = \frac{\sqrt{ (M-1) (2 \ln 2)^{d+1}}}{M}.
	\end{equation}
	\hfill\IEEEQED
	
	\subsection{Discussion on the Bound Optimality}\label{discussion}
	
	To investigate the bound tightness we compute and plot in Figs.~\ref{fig:z14}, \ref{fig:z13},  and \ref{fig:z5} the sequence $\mathbf{p}_d$ of pmf supported on a finite additive group $\mathcal{G}$ given by a fixed pmf $\mathbf{p}_0$ and the equation
	$\mathbf{p}_{d+1} = \mathbf{p}_d * \mathbf{p}_0$ where $*$ is the convolution with respect to the group $\mathcal{G}$. In other words, $\mathbf{p}_d$ is the pmf of the sum of $d+1$ i.i.d. $\mathcal{G}$-valued random variables with a law given by $\mathbf{p}_0$.
	
	Figs.~\ref{fig:z14} and~\ref{fig:z5} show that the presented bound is tight in two situations:
	\begin{enumerate}
		\item When the support of the random variables is in the coset of a sub-group of order $k+1$ the inequality is tight. This is the case in Fig.~\ref{fig:z14} as $\{\bar{0};\bar{7}\}$ is a finite sub-group of $\mathbb{Z}_{14}$ with two elements.
		\item In the high entropic regime, $k=M-1$ and there is always a sub-group, the group itself. This is the case in Fig.~\ref{fig:z5}.
	\end{enumerate}
	However, when there is no finite sub-group of order $k+1$ the inequality can be strictly violated as shown by Fig.~\ref{fig:z13}. Figs.~\ref{fig:z13} and~\ref{fig:z14} differs only by their group structure changed from $\mathbb{Z}_{14}$ to $\mathbb{Z}_{13}$, though the effect is huge on the actual entropy of the sum. Indeed $\{\bar{0};\bar{7}\}$ is not the coset of a sub-group of $\mathbb{Z}_{13}$, it is even spanning the whole group. As reported in \cite{madiman2021entropy} the Cauchy-Davenport inequality shows that for $A,B$ two subsets of $\mathbb{Z}_p$ (p prime), $|A+B| \geq \min \{ |A|+|B|-1, p \}$. As a consequence, the support of the sum must spread to the whole group very quickly. The investigation of this results may improve the presented inequalities.
	
	In Fig.~\ref{fig:z5}, we observe that the min-entropy does not increase visibly neither from $d=0$ to $d=1$ nor from $d=2$ to $d=3$.
	This supports the observation of Theorem~\ref{thm:refined-extension} that masking with odd order might not be relevant with respect to the worst case leakages as measured by the min-entropy.
	
	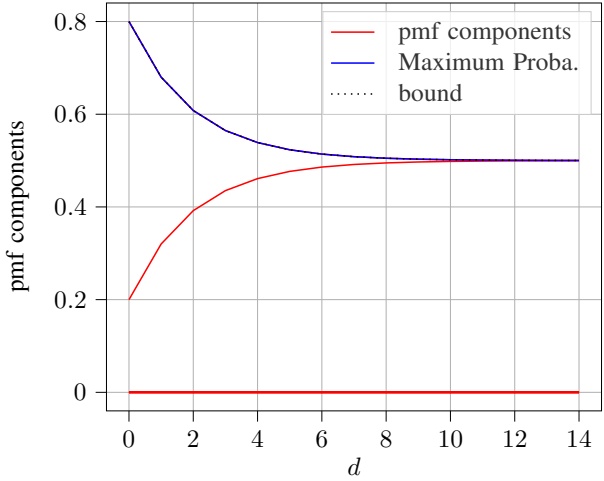
\begin{figure}[htbp!]
\begin{tikzpicture}[scale=0.95]

\definecolor{darkgray176}{RGB}{176,176,176}
\definecolor{lightgray204}{RGB}{204,204,204}

\begin{axis}[
legend cell align={left},
legend style={fill opacity=0.8, draw opacity=1, text opacity=1, draw=lightgray204},
tick align=outside,
tick pos=left,
x grid style={darkgray176},
xlabel={\(\displaystyle d\)},
xmajorgrids,
xmin=-0.7, xmax=14.7,
xtick style={color=black},
y grid style={darkgray176},
ylabel={pmf components},
ymajorgrids,
ymin=-0.04, ymax=0.84,
ytick style={color=black}
]
\addplot [semithick, red]
table {%
0 0.8
1 0.68
2 0.608
3 0.5648
4 0.53888
5 0.523328
6 0.5139968
7 0.50839808
8 0.505038848
9 0.5030233088
10 0.50181398528
11 0.501088391168
12 0.5006530347008
13 0.50039182082048
14 0.500235092492288
};
\addlegendentry{pmf components}
\addplot [semithick, red, forget plot]
table {%
0 0
1 0
2 0
3 0
4 0
5 0
6 0
7 0
8 0
9 0
10 0
11 0
12 0
13 0
14 0
};
\addplot [semithick, red, forget plot]
table {%
0 0
1 0
2 0
3 0
4 0
5 0
6 0
7 0
8 0
9 0
10 0
11 0
12 0
13 0
14 0
};
\addplot [semithick, red, forget plot]
table {%
0 0
1 0
2 0
3 0
4 0
5 0
6 0
7 0
8 0
9 0
10 0
11 0
12 0
13 0
14 0
};
\addplot [semithick, red, forget plot]
table {%
0 0
1 0
2 0
3 0
4 0
5 0
6 0
7 0
8 0
9 0
10 0
11 0
12 0
13 0
14 0
};
\addplot [semithick, red, forget plot]
table {%
0 0
1 0
2 0
3 0
4 0
5 0
6 0
7 0
8 0
9 0
10 0
11 0
12 0
13 0
14 0
};
\addplot [semithick, red, forget plot]
table {%
0 0
1 0
2 0
3 0
4 0
5 0
6 0
7 0
8 0
9 0
10 0
11 0
12 0
13 0
14 0
};
\addplot [semithick, red, forget plot]
table {%
0 0.2
1 0.32
2 0.392
3 0.4352
4 0.46112
5 0.476672
6 0.4860032
7 0.49160192
8 0.494961152
9 0.4969766912
10 0.49818601472
11 0.498911608832
12 0.4993469652992
13 0.49960817917952
14 0.499764907507712
};
\addplot [semithick, red, forget plot]
table {%
0 0
1 0
2 0
3 0
4 0
5 0
6 0
7 0
8 0
9 0
10 0
11 0
12 0
13 0
14 0
};
\addplot [semithick, red, forget plot]
table {%
0 0
1 0
2 0
3 0
4 0
5 0
6 0
7 0
8 0
9 0
10 0
11 0
12 0
13 0
14 0
};
\addplot [semithick, red, forget plot]
table {%
0 0
1 0
2 0
3 0
4 0
5 0
6 0
7 0
8 0
9 0
10 0
11 0
12 0
13 0
14 0
};
\addplot [semithick, red, forget plot]
table {%
0 0
1 0
2 0
3 0
4 0
5 0
6 0
7 0
8 0
9 0
10 0
11 0
12 0
13 0
14 0
};
\addplot [semithick, red, forget plot]
table {%
0 0
1 0
2 0
3 0
4 0
5 0
6 0
7 0
8 0
9 0
10 0
11 0
12 0
13 0
14 0
};
\addplot [semithick, red, forget plot]
table {%
0 0
1 0
2 0
3 0
4 0
5 0
6 0
7 0
8 0
9 0
10 0
11 0
12 0
13 0
14 0
};
\addplot [semithick, blue]
table {%
0 0.8
1 0.68
2 0.608
3 0.5648
4 0.53888
5 0.523328
6 0.5139968
7 0.50839808
8 0.505038848
9 0.5030233088
10 0.50181398528
11 0.501088391168
12 0.5006530347008
13 0.50039182082048
14 0.500235092492288
};
\addlegendentry{Maximum Proba.}
\addplot [semithick, black, dotted]
table {%
0 0.8
1 0.68
2 0.608
3 0.5648
4 0.53888
5 0.523328
6 0.5139968
7 0.50839808
8 0.505038848
9 0.5030233088
10 0.50181398528
11 0.501088391168
12 0.5006530347008
13 0.50039182082048
14 0.500235092492288
};
\addlegendentry{bound}
\end{axis}

\end{tikzpicture}
		\caption{Probability mass function of the sum of $d+1$ i.i.d. $\mathbb{Z}_{14}$-valued random variables with probability mass function $\mathbf{p}_0$ with $p_0=0.8$ and $p_7=0.2$.}
		\label{fig:z14}
	\end{figure}
	
	\begin{figure}[htbp!]
\begin{tikzpicture}[scale=0.95]

\definecolor{darkgray176}{RGB}{176,176,176}
\definecolor{lightgray204}{RGB}{204,204,204}

\begin{axis}[
legend cell align={left},
legend style={fill opacity=0.8, draw opacity=1, text opacity=1, draw=lightgray204},
tick align=outside,
tick pos=left,
x grid style={darkgray176},
xlabel={\(\displaystyle d\)},
xmajorgrids,
xmin=-0.7, xmax=14.7,
xtick style={color=black},
y grid style={darkgray176},
ylabel={pmf components},
ymajorgrids,
ymin=-0.04, ymax=0.84,
ytick style={color=black}
]
\addplot [semithick, red]
table {%
0 0.8
1 0.64
2 0.512
3 0.4096
4 0.32768
5 0.262144
6 0.2097152
7 0.16777216
8 0.134217728
9 0.1073741824
10 0.0858993459200001
11 0.068719476736
12 0.054975582208
13 0.04398047428608
14 0.035184427139072
};
\addlegendentry{pmf components}
\addplot [semithick, red, forget plot]
table {%
0 0
1 0.04
2 0.096
3 0.1536
4 0.2048
5 0.24576
6 0.2752512
7 0.29360128
8 0.301989888
9 0.301989888
10 0.2952790016
11 0.283467841536
12 0.2680059592704
13 0.25013889531904
14 0.230897441865728
};
\addplot [semithick, red, forget plot]
table {%
0 0
1 0
2 0
3 0.0016
4 0.0064
5 0.01536
6 0.028672
7 0.0458752
8 0.0660602880000001
9 0.0880803840000001
10 0.1107296256
11 0.13287555072
12 0.153545080832
13 0.17197049053184
14 0.18760417148928
};
\addplot [semithick, red, forget plot]
table {%
0 0
1 0
2 0
3 0
4 0
5 6.4e-05
6 0.0003584
7 0.00114688
8 0.002752512
9 0.005505024
10 0.00968884224000001
11 0.015502147584
12 0.0230317621248
13 0.03224446697472
14 0.04299262263296
};
\addplot [semithick, red, forget plot]
table {%
0 0
1 0
2 0
3 0
4 0
5 0
6 0
7 2.56e-06
8 1.8432e-05
9 7.3728e-05
10 0.0002162688
11 0.00051904512
12 0.0010796138496
13 0.00201527918592
14 0.00345476431872
};
\addplot [semithick, red, forget plot]
table {%
0 0
1 0
2 0
3 0
4 0
5 0
6 0
7 0
8 0
9 1.024e-07
10 9.0112e-07
11 4.325376e-06
12 1.49946368e-05
13 4.198498304e-05
14 0.000100763959296
};
\addplot [semithick, red, forget plot]
table {%
0 0
1 0
2 0
3 0
4 0
5 0
6 0
7 0
8 0
9 0
10 0
11 4.096e-09
12 4.25984e-08
13 2.3855104e-07
14 9.54204160000001e-07
};
\addplot [semithick, red, forget plot]
table {%
0 0.2
1 0.32
2 0.384
3 0.4096
4 0.4096
5 0.393216
6 0.3670016
7 0.33554432
8 0.301989888
9 0.268435456
10 0.23622320128
11 0.206158430208
12 0.1786706395136
13 0.15393162805248
14 0.1319413972992
};
\addplot [semithick, red, forget plot]
table {%
0 0
1 0
2 0.008
3 0.0256
4 0.0512
5 0.08192
6 0.114688
7 0.14680064
8 0.176160768
9 0.201326592
10 0.2214592512
11 0.23622320128
12 0.2456721293312
13 0.25013889531904
14 0.25013889531904
};
\addplot [semithick, red, forget plot]
table {%
0 0
1 0
2 0
3 0
4 0.00032
5 0.001536
6 0.0043008
7 0.00917504000000001
8 0.016515072
9 0.0264241152
10 0.03875536896
11 0.0531502202880001
12 0.0690952863744001
13 0.0859852452659201
14 0.103182294319104
};
\addplot [semithick, red, forget plot]
table {%
0 0
1 0
2 0
3 0
4 0
5 0
6 1.28e-05
7 8.192e-05
8 0.000294912
9 0.000786432000000001
10 0.0017301504
11 0.003321888768
12 0.00575794053120001
13 0.00921270484992001
14 0.01381905727488
};
\addplot [semithick, red, forget plot]
table {%
0 0
1 0
2 0
3 0
4 0
5 0
6 0
7 0
8 5.12e-07
9 4.096e-06
10 1.80224e-05
11 5.767168e-05
12 0.000149946368
13 0.00033587986432
14 0.000671759728640001
};
\addplot [semithick, red, forget plot]
table {%
0 0
1 0
2 0
3 0
4 0
5 0
6 0
7 0
8 0
9 0
10 2.048e-08
11 1.96608e-07
12 1.0223616e-06
13 3.81681664e-06
14 1.145044992e-05
};
\addplot [semithick, blue]
table {%
0 0.8
1 0.64
2 0.512
3 0.4096
4 0.4096
5 0.393216
6 0.3670016
7 0.33554432
8 0.301989888
9 0.301989888
10 0.2952790016
11 0.283467841536
12 0.2680059592704
13 0.25013889531904
14 0.25013889531904
};
\addlegendentry{Maximum Proba.}
\addplot [semithick, black, dotted]
table {%
0 0.8
1 0.68
2 0.608
3 0.5648
4 0.53888
5 0.523328
6 0.5139968
7 0.50839808
8 0.505038848
9 0.5030233088
10 0.50181398528
11 0.501088391168
12 0.5006530347008
13 0.50039182082048
14 0.500235092492288
};
\addlegendentry{bound}
\end{axis}

\end{tikzpicture}
		\caption{Probability mass function of the sum of $d+1$ i.i.d. $\mathbb{Z}_{13}$-valued random variables with probability mass function $\mathbf{p}_0$  with $p_0=0.8$ and $p_7=0.2$.}
		\label{fig:z13}
	\end{figure}
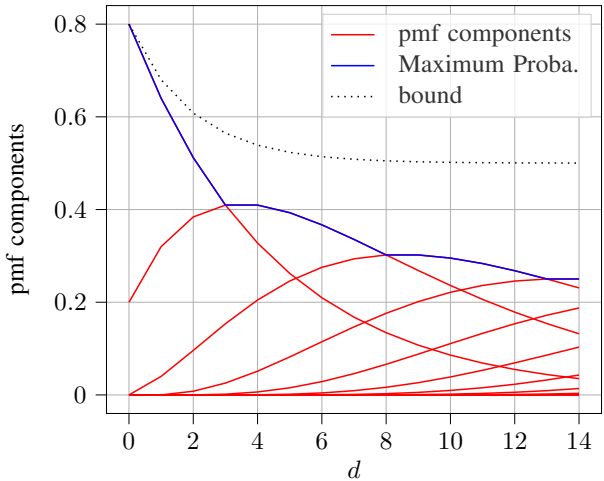
	
	\begin{figure}[htbp!]
\begin{tikzpicture}[scale=0.95]

\definecolor{darkgray176}{RGB}{176,176,176}
\definecolor{lightgray204}{RGB}{204,204,204}

\begin{axis}[
legend cell align={left},
legend style={fill opacity=0.8, draw opacity=1, text opacity=1, draw=lightgray204, at={(0.02,0.02)},anchor=south west},
tick align=outside,
tick pos=left,
x grid style={darkgray176},
xlabel={\(\displaystyle d\)},
xmajorgrids,
xmin=-0.25, xmax=5.25,
xtick style={color=black},
y grid style={darkgray176},
ylabel={pmf components},
ymajorgrids,
ymin=-0.0125, ymax=0.2625,
ytick style={color=black}
]
\addplot [semithick, red]
table {%
0 0.25
1 0.1875
2 0.203125
3 0.19921875
4 0.19921875
5 0.199951171875
};
\addlegendentry{pmf components}
\addplot [semithick, red, forget plot]
table {%
0 0.25
1 0.1875
2 0.203125
3 0.203125
4 0.2001953125
5 0.199951171875
};
\addplot [semithick, red, forget plot]
table {%
0 0.25
1 0.1875
2 0.1875
3 0.19921875
4 0.2001953125
5 0.199951171875
};
\addplot [semithick, red, forget plot]
table {%
0 0.25
1 0.25
2 0.203125
3 0.19921875
4 0.2001953125
5 0.199951171875
};
\addplot [semithick, red, forget plot]
table {%
0 0
1 0.1875
2 0.203125
3 0.19921875
4 0.2001953125
5 0.2001953125
};
\addplot [semithick, blue]
table {%
0 0.25
1 0.25
2 0.203125
3 0.203125
4 0.2001953125
5 0.2001953125
};
\addlegendentry{Maximum Proba.}
\addplot [semithick, black, dotted]
table {%
0 0.25
1 0.25
2 0.203125
3 0.203125
4 0.2001953125
5 0.2001953125
};
\addlegendentry{bound}
\end{axis}

\end{tikzpicture}
		\caption{Probability mass function of the sum of $d+1$ i.i.d. $\mathbb{Z}_{5}$-valued random variables with probability mass function $\mathbf{p}_0 = (\tfrac{1}{4},\tfrac{1}{4},\tfrac{1}{4},\tfrac{1}{4},0)$.}
		\label{fig:z5}
	\end{figure}
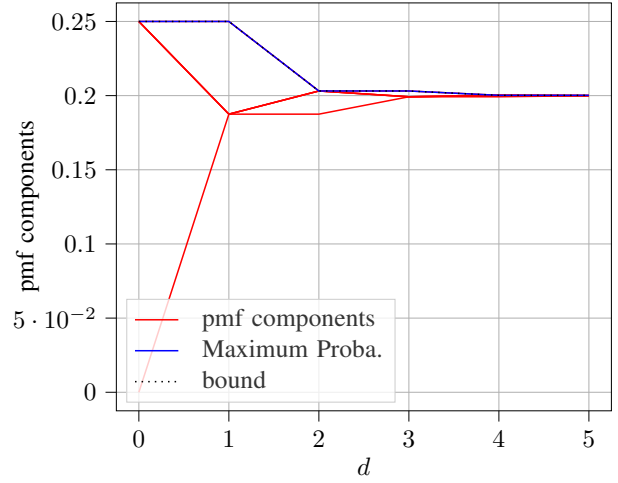

	\subsection{Bound Comparison}\label{toy-example} 
	
	Figs~\ref{fig:d1} and~\ref{fig:d2} compare both bounds for the toy example introduced. {\color{blueR} Though the bound obtained with $I_\infty$ does not change significantly from $d=1$ to $d=2$. The new bound performs better for this leakage. Especially for moderate (less than $10^5$) number of traces. The main limitation of this bound is that it relies on a high noise assumption.}  
	
	\begin{figure}[htbp!]
		\input{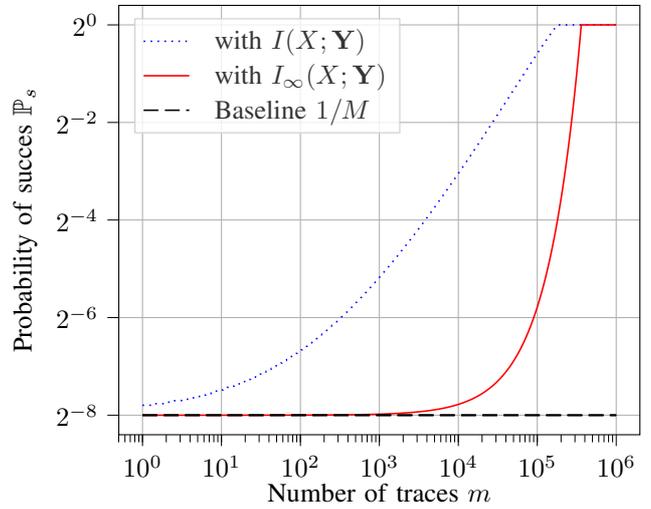}
		\caption{Comparison of the two upper bounds (ours, Corollary~\ref{cor-bound_PS}, versus state-of-the-art, namely \cite[Eqn.~8]{DBLP:journals/iacr/BeguinotCGLMRS22}) for $d=1$ and $M=256$}
		\label{fig:d1}
	\end{figure}
	
	\begin{figure}[htbp!]
		\input{ISITvsCOSADE_2}
		\caption{Comparison of the two upper bounds (ours, Corollary~\ref{cor-bound_PS}, versus state-of-the-art, namely \cite[Eqn.~8]{DBLP:journals/iacr/BeguinotCGLMRS22}) for $d=2$ and $M=256$}
		\label{fig:d2}
	\end{figure}
	
\end{appendix}

\end{document}